\documentclass[10pt]{amsart}
\usepackage{enumerate,amsmath,amssymb,latexsym,
amsfonts, amsthm, amscd}

\usepackage{enumerate,amsmath,amsthm,amssymb,latexsym,
amsfonts,amscd}
\usepackage[all]{xy}

\title[An observer principle for general relativity]{An observer principle for general relativity}

\author[Maurice J. Dupr\'e]{M\lowercase{aurice} J. D\lowercase{upr\'e}\\D\lowercase{epartment of} M\lowercase{athematics}
\\N\lowercase{ew} O\lowercase{rleans}, LA 70118\\\lowercase{email:  mdupre@tulane.edu}\\15 J\lowercase{anuary} 2014}
\address{DEPARTMENT OF MATHEMATICS\\TULANE UNIVERSTIY\\NEW ORLEANS, LA 70118}
\email{mdupre@tulane.edu}

\theoremstyle{plain}

\newtheorem{proposition}{Proposition}[section]
\newtheorem{theorem}{Theorem}[section]
\newtheorem{corollary}{Corollary}[section]
\newtheorem{postulate}{Postulate}[section]

\theoremstyle{definition}

\numberwithin{equation}{section}

\newcommand{\B}{\mathcal B}
\newcommand{\Ce}{\mathcal C}
\newcommand{\D}{\mathcal D}

\newcommand{\F}{\mathcal F}

\newcommand{\J}{\mathcal J}
\newcommand{\K}{\mathcal K}

\newcommand{\M}{\mathcal M}

\newcommand{\R}{\mathcal R}
\newcommand{\T}{\mathcal T}
\newcommand{\U}{\mathcal U}

\newcommand{\W}{\mathcal W}

\newcommand{\bC}{\mathbb{C}}

\newcommand{\bK}{\mathbb{K}}
\newcommand{\bN}{\mathbb{N}}

\newcommand{\bR}{\mathbb{R}}

\newcommand{\ra}{\rightarrow}

\newcommand{\lra}{\longrightarrow}

\newcommand{\del}{\partial}

\newcommand{\med}{\medbreak}

\begin{document}

\maketitle

\begin{abstract}

We give a mathematical uniqueness theorem which in particular shows that symmetric tensors  in general relativity are uniquely determined by their monomial functions on the light cone.  Thus, for an observer to observe a tensor at an event in general relativity is to contract with the velocity vector of the observer, repeatedly to the rank of the tensor.  Thus two symmetric tensors observed to be equal by all observers at a specific event are necessarily equal at that event.

\end{abstract}

\med \textbf{Mathematics Subject Classification (2000)} : 83C05,
83C40, 83C99.

\med \textbf{Keywords} : Gravity, general relativity, Einstein
equation, energy density.

\section{\bf INTRODUCTION}

\med

The mathematical formulation of observation in general relativity almost always involves the contraction of the observers (four) velocity vector with a tensor which captures mathematically what the observer wishes to observe.  Thus, at a specific event, if $T$ is the energy momentum  stress tensor of all matter and fields, then $T(u,u)=T_{\alpha \beta}u^{\alpha}u^{\beta}$ is the energy density observed by an observer with velocity $u$ at that specific event.  Suppose both $S$ and $T$ are candidates to be the energy momentum stress tensor of matter and fields, at a specific event.  Of course that means they are, in particular, symmetric tensors of the second rank. Then the principles of relativity should guarantee that if $S(u,u)=T(u,u)$ for every velocity vector in the forward light cone at that specific event, then $S=T.$  In fact this is the case as observed by Sachs and Wu \cite{SACHSWU}, and the result follows fairly simply from the symmetry of the two tensors.

We might say that in general, the laws of physics in general relativity are given by tensor equations of the form $S=T,$ and the justification of the equality of two rank $r$  symmetric tensors in general relativity would by the principle of relativity be that all observers observe them to be the same, which is to say that for every observer velocity vector $u$ in the light cone at the event, we have

\begin{equation}\label{monomial1}
S(u,u,u,\cdot \cdot \cdot ,u)=T(u,u, \cdot \cdot \cdot ,u).
\end{equation}

We will see that mathematically, if (\ref{monomial1}) holds for every velocity vector in the light cone, at a specific event, then in fact we can conclude $S=T$ at that specific event, provided that $S$ and $T$ are symmetric.   More generally, we need to consider equations formed by integrating contractions of rank $r$ symmetric tensor fields with $r$ velocity fields, and as the space of vector fields on a manifold is a vector space, we generally obtain real valued functions which are linear in each input vector field.  If covariant derivatives are involved, the resulting functions may not be linear over the smooth functions as coefficients, but merely linear over the constants.  In the most general setting, we would then have simply a vector space $V$ and a multilinear function $T$ of $r$ variables 

$$T: V^r \lra  \bR,$$
where
$$V^r=V \times V \times V \times \cdot \cdot \cdot \times V$$
is the $r-$fold cartesian product of the vector space $V$ with itself.   If $S$ and $T$ are both symmetric multilinear functions of rank $r$ on $V$ and if (\ref{monomial1}) holds for all $u$ in the subset $U$ of $V$, we would like to be able to conclude that $S=T.$  Obviously, there must be some restriction on what $U$ can be for this to work, and we would like to find a general condition which works in any vector space, without having to deal with topology.


\section{\bf THE OBSERVER PRINCIPLE}

\med

If the smooth manifold $M$ is a model for space-time of dimension $n+1$ in general relativity, then for fixed $m \in M,$ mathematically, $T_mM$ is a Lorentz vector space of dimension
$n+1,$ so taking any time-like unit vector, say $u,$ and
defining $g_u(v,w)=2g(u,v)g(u,w)+g(v,w)$ gives a Euclidean metric
on $T_mM$ making it in particular into a Banach space of finite
dimension. Thus, $T_mM$ is an example of a Banachable space-a
topological vector space whose topology can be defined by a norm.
This topology is actually well-known to be independent of the
choice of $u$ in case of finite dimensions. In fact, any finite dimensional vector space has a unique topology making it a topological vector space \cite{JLKELLEY2}.  Differential geometry
can be easily based on such spaces, \cite{UPMEIER}, and for some examples in
infinite dimension, the interested reader can see
\cite{KRIGL&MICHOR}, \cite{BELTITA}, \cite{DUPGLAZEPREV}, \cite{DUPREGLAZE1} and
\cite{DUPREGLAZE2}. In particular, the theory of analytic
functions and power series all goes through for general Banachable
spaces \cite{UPMEIER}.  We would like to point out how this can be applied to the
theory of Lorentz vector spaces and vector spaces of vector fields on spacetimes.

In general, suppose that $E_1,E_2,\cdot \cdot \cdot , E_r,$ and $F$ are all vector spaces
(possibly infinite dimensional and not necesssarily topological). Recall the function or mapping
$$A: E_1 \times E_2 \times \cdot \cdot \cdot \times E_r \lra F$$
is a {\it multilinear
map} provided that it is linear in each variable when all others
are held fixed, and in this case, we say that $A$ is a multilinear map
of rank $r$ on $E_1 \times E_2 \times...\times E_r$ with values in $F.$  A useful notation here is just to use juxtaposition
for evaluation of multilinear maps, so we write
$$A(v_1,v_2, \cdot \cdot \cdot ,v_r)=Av_1v_2 \cdot \cdot \cdot v_r$$ 
whenever $v_k \in E_k$ for $1
\leq k \leq r.$  Thus, we simply treat the multilinear map $A$ as
a sort of generalized coefficient which allows us to multiply
vectors, and the multilinear condition simply becomes the
distributive law of multiplication.

In case that $E_k=E$ for all $k,$ there is really a single vector
space providing the input vectors, and $A:E^r \lra F.$ We say that
$A$ is a multilinear map of rank $r$ on $E$ in this case, even
though in reality, the domain of $A$ is the set $E^r.$ Here it is
useful to write $v^{(k)}$ for the $k-$fold juxtaposition of $v$'s.
Thus we have
$$A(v,v,...,v)=Av^{(r)}.$$ 
More generally, then for any positive
integer $m$ and vectors $v_1,v_2,...,v_m \in E$ and non-negative
integers $k_1,k_2,...,k_m$ satisfying $k_1+k_2+...+k_m=r,$ we have
the equation
$$Av_1^{(k_1)}v_2^{(k_2)} \cdot \cdot \cdot v_m^{(k_m)}=A(v_1,\cdot \cdot \cdot ,v_1,v_2,\cdot \cdot \cdot,v_2,\cdot \cdot \cdot,v_m,\cdot \cdot \cdot ,v_m)$$ 
where
each vector is repeated the appropriate number of times, $v_1$
being repeated $k_1$ times, $v_2$ repeated $k_2$ times and so on.
If $k_i=0$ then that merely means that $v_i$ is
actually left out, so $v^{(0)}=1$ in effect.

Of course, we say that $A:E^r \lra F$ is symmetric if $Av_1v_2...v_r$ is
independent of the ordering of the $r$ input vectors.  Thus when
dealing with algebraic expressions involving symmetric multilinear
maps as coefficients, the commutative law is in effect.  We can define
the {\it monomial} function $f_A:E \lra F$ by the rule
$f_A(x)=A(x,x,x,\cdot \cdot \cdot ,x)=Ax^{(r)}.$

We denote by $L(E_1,E_2,\cdot \cdot \cdot ,E_r;F)$ the vector space of all multilinear maps of $E_1 \times E_2 \times \cdot \cdot \cdot \times E_r$ into $F,$ and set $L^r(E;F)=L(E_1,E_2,\cdot \cdot \cdot ,E_r;F)$ when all $E_k$ are the same vector space $E.$  Of course, $L^1(E;F)=L(E;F)$ is just the vector space of all linear maps from $E$ to $F.$  We denote the dual space of $E$ by $E^*=L(E;\bR).$  We use $L^r_{sym}(E;F)$ to denote the vector subspace of $L^r(E;F)$ consisting of the symmetric multilinear maps.  There is a natural isomorphism
$$L(E_1,E_2, \cdot \cdot \cdot ,E_r;F)\cong L(E_1, \cdot \cdot \cdot ,E_{r-1};L(E_r;F))$$ 
which identifies the $F-$valued rank $r$ multilinear map $A$ with the $L(E_r;F)-$valued multilinear map $B$ of rank $r-1$ given by
$$[Bv_1v_2 \cdot \cdot \cdot v_{r-1}](v_r)=Av_1v_2 \cdot \cdot \cdot v_r,$$  
and notice that if $A$ is symmetric then so is $B,$ but of course the converse may not be true.  In any case, it is useful to simply denote $Bv_1v_2 \cdot \cdot \cdot v_{r-1}=Av_1v_2\cdot \cdot \cdot v_{r-1}$ in this situation.

Suppose now that $E$ and $F$ are
any vector spaces and $A$ is a symmetric multilinear
map (tensor) on $E$ with values in $F,$ of rank $r.$  We will begin for simplicity by restricting to Banachable spaces, that is topolgical vector spaces whose topology is complete and comes from a norm.  In addition for simplicity, we assume that $A$ is continuous.  For in case $E$ and $F$ are Banachable spaces and $A$ is continuous, $f_A$ is an analytic
function. In fact, if $x_1,x_2,x_3, \cdot \cdot \cdot ,x_r \in E,$ then
differentiating, using proposition 3.3 and repeated application of
propositions 3.5 and 3.8 of \cite{LANG}, page 10, we find

\begin{equation}\label{analyticcontinuation}
D_{x_1}D_{x_2}D_{x_3}  \cdot \cdot \cdot D_{x_r}f_A(a)=(n!)A(x_1,x_2,x_3, \cdot \cdot \cdot  ,x_r),~~a
\in E.
\end{equation}
From (\ref{analyticcontinuation}), we see very generally that if
$U$ is any open subset of $E$ on which $f_A$ is constant, then in
fact, $A=0,$ since we can choose $a \in U.$ Indeed, if $a \in U,$
since $f_A$ is constant on $U,$ it follows that the derivative on
the left side of the equation (\ref{analyticcontinuation}) is 0,
and hence the right side is 0, for every possible choice of
vectors $x_1,x_2,x_3,  \cdot \cdot \cdot , x_r \in E.$ But notice that $a$ does not
appear on the right hand side of (\ref{analyticcontinuation}),
only $A(x_1,x_2,x_3,  \cdot \cdot \cdot , x_r),$ and the vectors $x_1,x_2,x_3,  \cdot \cdot \cdot ,  x_r$
can be chosen arbitrarily. Thus, $A=0$ follows. We have therefore
proven a special case of the following mathematical theorem, for which the proof in general will be given after some remarks.

\begin{theorem}\label{gen top anal cont}
Suppose $E$ is any topological vector space and $F$ is any vector
space. Suppose $A:E^r \lra F$ and $B:E^r \lra F$ are any symmetric multilinear
maps of rank $r.$ If there is a non-empty open subset of $E$ on
which $f_A-f_B:E \lra F$ is constant, then $A=B.$
\end{theorem}
We emphasize that the vector spaces here may be infinite dimensional and the multilinear map need not be continuous.

Notice that Theorem \ref{gen top anal cont} is a well known special case of
the uniqueness of general power series (there is only one term
here). For a purely algebraic proof in the case $r=2,$ which is
the case of most importance here, we refer the interested reader
to \cite{DUPRE}.  See also page 72 of \cite{SACHSWU} or page 260 of \cite{KRIELE} for a proof using
differentiation for the special case $r=2$ which is similar in form to
that given here next.  As well, the result for $r=2$ can easily be proved directly using algebra alone by the technique of polarization as in \cite{DUPRE}.

\begin{corollary}\label{obsv2} {\bf OBSERVER PRINCIPLE.} If $A$
and $B$ are both symmetric tensors of rank $r$ on $T_mM$ with
values in $F,$ and if $Au^{(r)}=Bu^{(r)}$ for every time-like unit
vector in $T_mM,$ then $A=B.$
\end{corollary}

\begin{proof} Since $f_A$ and $f_B$ are homogeneous functions of
degree $r,$ it follows that the hypothesis guarantees
$Av^{(r)}=Bv^{(r)}$ for all $v$ in the light cone of $T_mM$ which
is an open subset of $T_mM.$
\end{proof}

If we define $U(T_mM)$ to be the set of time-like unit
vectors in $T_mM,$ then this set has a topology called the
relative topology as a subset of $T_mM$ and we have a retraction
function given by normalization which retracts the light cone onto
$U(T_mM).$ It follows immediately that if $W$ is any (relatively)
open subset of $U(T_mM),$ then the hypothesis of the observer
principle can be weakened to merely require $Au^{(r)}=Bu^{(r)}$
for each $u \in W.$ In particular, if we choose a time orientation
on $T_mM,$ then we can merely require $Au^{(r)}=Bu^{(r)}$ for each
future time-like unit vector in $T_mM.$ This is in a sense, the
essence of the {\bf Principle of Relativity}, for instance, as
applied to second rank symmetric tensors-a law (at $m$), say
$A=B,$ should be true for all observers (at $m$) and conversely,
if true for all observers (at $m$), that is if $A(u,u)=B(u,u)$ for
all (future) time-like unit vectors $u \in T_mM,$ then it should
be a law (at $m$) that $A=B.$ It is for this reason that we call
Corollary \ref{obsv2} the observer principle.

We wish to be able to apply the observer principle to multilinear maps defined on vector spaces of sections of tensor bundles given by integration of sections, so we need the complete generality of Theorem \ref{gen top anal cont} whose proof we turn to now.  In fact, because of this need, it will be useful to be even more general, but the statement of the theorem becomes slightly more technical.  To give the general statement, we need to define the {\it star} of a subset of a vector space.  If $U \subset E,$ its star, denoted $Star(U,E)$ is the set of all points $v \in U$ having the property that for each $w \in E$ there is some positive number  $\delta_w$ so that $v+tw \in U$ for any number $t$ with $|t| \leq \delta_w.$  If $E$ is a topological vector space, then each open subset is equal to its star.  Set $Star^1(U,E)=Star(U,E)$ and inductively, define $$Star^r(U,E)=Star^{r-1}(Star(U,E)),$$ and call this the $r-$star of $U$ in $E.$  For instance, if $Star(U,E)=U,$ then obviously $Star^r(U,E)=U,$ for every $r\geq 1.$

We will also need to use the fact that if $v \in F,$ then there is $\lambda \in F^*$ with $\lambda(v)\neq 0.$  This fact in general vector spaces requires the existence of a spanning linearly independent set which is guaranteed by the Axiom of Choice of set theory.  As such an axiom might be objectionable in applications to physics, we circumvent this by simply defining $F$ to be {\it non-degenerate} provided that for each vector $v$ in $F$ there is a member $\lambda$ of $F^*$ with $\lambda(v) \neq 0.$  In all applications to vector spaces of tensor fields in physics, the non-degeneracy is usually obvious.  However, the topology is usually not, so the notion of the star will circumvent the need to actually deal with topological vector spaces, beyond merely noting that each open subset of a topological vector space equals its star and therefore its $r-$star for all $r \geq 1.$

\begin{theorem}\label{genanalcont} Suppose $F$ is a non-degenerate vector space. If $A$ and $B$ are a symmetric multilinear maps of rank $r$ on a vector space $E$ with values in the vector space $F$ and if the monomial
function $f_A-f_B:E \lra F$ is constant on a set having non-empty $r-$star, then $A=B.$
\end{theorem}

Obviously Theorem \ref{gen top anal cont} is a consequence of Theorem \ref{genanalcont}, by our previous remarks. Clearly, as $f_{A-B}=f_A-f_B,$ it suffices to prove the case with $A$ arbitrary and $B=0.$  Thus we begin by assuming that $A$ is an arbitrary symmetric multilinear map of rank $r$ on the arbitrary vector space $E$ with values in the arbitrary non-degenerate vector space $F.$

Then for any $v_0,v_1, \cdot \cdot \cdot  v_m \in E,$

\begin{equation}\label{multnomthm}
f_A(v_0+v_1+\cdot \cdot \cdot+v_m)=\sum_{[k_0+k_1+...k_m=r]}C(r;k_0,k_1, \cdot \cdot \cdot ,k_m)Av_0^{(k_0)}v_1^{(k_1)} \cdot \cdot \cdot v_m^{(k_m)}.
\end{equation}
Here $C(r;k_0,k_1,  \cdot \cdot \cdot , k_m)$ is the multinomial coefficient:
\begin{equation}\label{multnomcoeff}
C(r;k_0,k_1, \cdot \cdot \cdot , k_m)=\frac{r!}{k_0!k_1!  \cdot \cdot \cdot  k_m!}.
\end{equation}

Now, proceeding inductively, let us notice that if $r=1,$ then the theorem is a triviality, since a linear map which is constant on any non-empty star is easily seen to be identically zero.  Assume the theorem is already proven for the case of rank $r-1.$  Suppose that $f_A$ is constant with value $C$ on $U \subset E$ with $Star^r(U,E)$ non-empty.  Choose $v_0 \in Star(U,E) \subset U,$ take any $w \in E,$ and any $\lambda \in F^*.$  Then choose $\delta >0$ such that $v_0+tw \in U,$ whenever $|t| \leq \delta.$  We have
$$C=f_A(v_0+tw)=\sum^r_{k=0}C(r;k,r-k)t^kAv_0^{(r-k)}w^{(k)}=f_A(v_0)+\sum^r_{k=1} C(r;k,r-k)t^kAv_0^{(r-k)}w^{(k)}.$$ 
Thus, since also $f_A(v_0)=C,$ we must in fact have
$$\sum^r_{k=1}C(r;k,r-k)t^kAv_0^{(r-k)}w^{(k)}=0,$$ 
for any number $t$ with $|t| \leq \delta.$

Applying $\lambda$ to the previous vanishing equation gives a real-valued polynomial function on $\bR$ of degree $r$ which vanishes for all $t$ with $|t| \leq \delta.$  Since such a polynomial function can have at most $r$ roots, this vanishing implies all coefficients are zero.  In particular, this means that $$\lambda(Av_0^{(r-1)}w)= 0.$$

Since $F$ is non-degenerate and $\lambda \in F^*$ was arbitrary, $Av_0^{(r-1)}w = 0$  must be the case.  That is, the linear map $Av_0^{(r-1)}=0.$  But, $v_0$ is an arbitrary point of $Star(U,E).$   This means that if we define the rank $r-1$ multilinear map $B$ by $[Bv_1v_2  \cdot \cdot \cdot v_{r-1}]w=Av_1v_2  \cdot \cdot \cdot  v_{r-1}w,$ then $f_B$ vanishes identically on $Star(U,E)$ which has non-empty $(r-1)-$star, and therefore by the inductive hypothesis, $B=0.$  But this obviously implies $A=0,$ and the proof is complete.  We take this opportunity to point out that the proof given in the appendix of \cite{DUPRE2} for this very general case is invalid, so the proof here provides a correction to that appendix, as well as a further generalization.

By convention, a multilinear map from $E$ to $F$ of rank zero is
just a vector in $F.$  If $A_k$ is a symmetric multilinear map of
$E$ to $F$ of rank $k,$ for $0 \leq k \leq r,$ then the function

$$f=\sum_{k=0}^r f_{A_k}$$
is a called a polynomial function of degree $r.$   If $U$ is a subset of $E$ having non-empty star and on which $f$ is constant, then we can for fixed $u$ in $U$ choose a linear function $g$ in $F^*$ with $g(f(u)) \neq 0$ and define $h: \bR \lra \bR$ by

$$h(t)=g(f(tu)),$$
and we see that $h$ is a real-valued polynomial function of a real variable which is constant on an infinite set and is therefore identically zero-that is, all its coefficients must be zero, and therefore $A_ku^{(k)}=0$ for each $k \leq r,$ so if $U$ has non-empty $r-$star, then $A_k=0$ for each $k \leq r.$


The general principle of analytic
continuation relies on the uniqueness of power series expressions.
In general, for Banach spaces, if two power series agree locally
as functions, then all their coefficients are the same-that is,
they are the same power series. The proof is easy using
differentiation, just use the same method used in freshman
calculus, but for Banach space valued functions. We have basically
proven this fact in case there are only a finite number of terms in the power
series, but without any topology required for the vector spaces involved.


\section{\bf VECTOR SPACES OF SMOOTH SECTIONS OF A SMOOTH VECTOR BUNDLE}

Our main application of the observer principle in infinite dimensions will be to vector spaces of smooth sections of $TM,$ so we would like to know that a symmetric multilinear map on the vector space of sections of $TM$ must vanish if its monomial form vanishes on all timelike vector fields.   More generally, we can consider any vector bundle $\xi$ over any manifold $M,$ and any semi-Riemannian metric on the vector bundle.  Thus, if we know that the set of timelike vector fields equals its star, then as it equals its $r-$star for all $r \geq 0,$ then the observer principle applies by our Theorem
\ref{genanalcont}.  The next proposition solves this problem for the case of vector spaces of vector fields which are continuous over a given fixed compact subset of $M.$

\begin{theorem}\label{timelikestar}
Suppose that $M$ is a smooth manifold and that $\xi$ is a smooth semi-Riemannian vector bundle over $M$ with metric tensor $g.$  Suppose that $K$ is a compact subset of $M$ and $E$ is a vector subspace of the set of all continuous sections of the vector bundle $\xi_M|K.$  Let $U$ be the set of all sections $v$ in $E$ satisfying $g(v,v) < 0$ on $K.$  Then $$Star(U,E)=U.$$
\end{theorem}

The proof of this theorem is a simple application of one of the most useful, simple, and beautiful theorems in point-set topology which is due to A. D. Wallace (my first mathematical mentor) \cite{JLKELLEY}.

\begin{theorem} ({\bf A. D.  Wallace.})  If $X$ and $Y$ are any topological spaces, if $A$ is a compact subset of $X$ and $B$ is a compact subset of $Y,$ and if $W$ is an open subset of $X\times Y$ which contains $A \times B,$ then there are open subsets $U$ of $X$ and $V$ of $Y,$ respectively, such that $A \subset U,~~B \subset Y,$ and $U \times V \subset W.$
\end{theorem}

It is customary to call a set of the form $A \times B$ a rectangle or to call it rectangular.  For a proof of Wallace's Theorem, we refer to \cite{JLKELLEY}, but it is an elementary exercise in topology.  Also, it is elementary in topology that $U \times V$ is open if and only if both $U$ and $V$ are open, whereas only slightly less elementary is the fact that $A \times B$ is compact if and only if both $A$ and $B$ are compact.  An open set which contains $A$ is said to be an open neighborhood of $A.$  Thus Wallace's Theorem says simply every open neighborhood of a compact rectangle contains a rectangular open neighborhood of that compact rectangle.

To use Wallace's Theorem here, given $v_0 \in U,$ and any $w \in E,$ we define the real valued function $f_w:K \times \bR \lra \bR$ by

$$f_w(m,t)=g(v_0(m)+tw(m),v_0(m)+tw(m)),~~(m,t) \in K \times \bR.$$  Since all vector fields in $E$ are assumed continuous, it follows that $f_w$ is continuous.  Also, clearly

$$f_w(K \times \{0\}) \subset N,$$ where $N$ denotes the set of all negative real numbers.  Thus, as $N$ is an open subset of $\bR$ and $f_w$ is continuous, it follows that its inverse image $W=f_w^{-1}(N)$ is an open subset of $K \times \bR$ and hence is an open neighborhood of the compact rectangle $K \times \{0\} \subset K \times \bR.$  Keeping in mind that $K$ is an open subset of itself, by Wallace's Theorem, there is an open rectangle $K \times V_w$ with $$K \times \{0\} \subset K \times V_w \subset W.$$  Thus, $V_w$ is an open neighborhood of zero in $\bR$ so there is a positive number $\delta_w$ with the property that if $|t| < \delta_w,$ then $t \in V_w.$  This means that $v_0 +tw \in U,$ for any $t \in V_w,$ and in particular, for any $t$ with $|t| < \delta_w.$  As $w$ was arbitrary in $E,$ this means, $v_0 \in Star(U,E).$  As $v_0$ was an arbitrary vector field in the set $U,$ it follows that $Star(U,E)=U$ as claimed.

Combining Theorem \ref{genanalcont} and Theorem \ref{timelikestar}, we then have immediately the final result on the observer principle.

\begin{theorem}\label{timelikeobserverprinciple}
If $K$ is a compact subset of $M$ and if $S$ and $T$ are symmetric multilinear maps of rank $r$ on a vector space $E$ of continuous vector fields on $K$ with the property that $Sv^r=Tv^r$ for every timelike vector field $v \in E,$ then $S=T.$
\end{theorem}
We merely need to observe that $E$ must be non-degenerate. Indeed, if $v \in E$ with $v \neq 0,$ then there is some particular $m \in K$ with $v(m) \neq 0,$ and then we can choose any $\lambda \in (T_mM)^*$ with $\lambda(v(m)) \neq 0,$ to obtain an element $f \in E^*$ with $f(v) \neq 0,$ namely $f=\lambda [ev_m],$ where $ev_m :E \lra T_mM$ is the evaluation map, $ev_m(w)=w(m),$ for all $w \in E.$

\section{\bf APPLICATIONS OF THE OBSERVER PRINCIPLE}

As an application of Theorem \ref{timelikeobserverprinciple}, we will apply it to integrals of operators which are more general than tensor fields.  Let us call $S$ a {\it tensor operator} of rank $r$ on $K \subset M$ provided that it is a multilinear map of rank $r$ on the vector space of smooth vector fields on $K$ and whose values are also smooth scalar fields on $K$ which has the property that if $m \in K,$ and if $f$ is a smooth function on $K$ which is constant in an open neighborhood of $m,$ then

$$S(v_1,v_2,\cdot \cdot \cdot,fv_l,\cdot \cdot \cdot,v_r)(m)=f(m)S(v_1,v_2,\cdot \cdot \cdot,v_r)(m).$$
 For instance, $S$ could be simply a tensor field of rank $r,$ but more generally, $S$ could be formed by covariant differentiation operators and tensor fields so as to be multilinear of rank $r,$ but not necessarily a tensor field of rank $r.$  Suppose that $\mu$ is a volume form on $M,$ and that $K \subset M.$   We can then define the scalar valued rank $r$ multilinear map $\hat{S}$ on the vector space $\Gamma_K$ of all smooth vector fields on $K$ by

\begin{equation}\label{integralS}
\hat{S}(w_1,w_2,\cdot \cdot \cdot,w_r)=\int_K S(w_1,w_2,\cdot \cdot \cdot,w_r)\mu,  \mbox{ for any }~~~~w_1,w_2,\cdot \cdot \cdot,w_r \in \Gamma_K.
\end{equation}
We will call $\hat{S}$ the integral of $S$ over $K,$ and denote it by

\begin{equation}\label{integraltensorfield}
\int_K S \mu =\hat{S},
\end{equation}
so we have

\begin{equation}\label{integraltensorfield2}
\left[\int_K S \mu \right](w_1,w_2, \cdot \cdot \cdot,w_r)=\int_K S(w_1,w_2,\cdot \cdot \cdot,w_r) \mu, \mbox{ for any }~~w_1,w_2, \cdot \cdot \cdot ,w_r \in \Gamma_K.
\end{equation}

Notice that the integral of such an operator is a multilinear map, which is therefore a special kind of tensor-not a tensor field.  By forming the multilinear map $\int S \mu,$ we in effect get around the problem that generally it does not make sense to integrate a tensor field itself, without choosing some kind of coordinate representation, as integrating components of a tensor does not give a tensor in the usual sense that physicists use the term.

It is clear that if $\int_U S \mu$ vanishes for every sufficiently small open subset of $U$ of $K,$ then $S$ itself must vanish.  For if $S(w_1,w_2,...,w_r)$ does not vanish at $m \in K,$ then it maintains its sign over some open neighborhood $U$ of $m,$ so the integral over $U$ would be either positive or negative but not zero.  To make the term sufficiently small precise here, we could say that if $\U$ is an open cover of $K$ and if $U$ is contained is some member of $\U,$ then $U$ is $\U-$small.  Then we say $U$ is sufficiently small if it is $\U-$small for some open cover $\U$ of $K.$

In any case, we see immediately that if $S$ is symmetric, then so is $\int S \mu.$  Thus if $S$ and $T$ are both symmetric, then to know that $S=T,$ by  Theorem \ref{timelikeobserverprinciple}, it is sufficient to know that the monomial forms of $\int S \mu$ and $\int T \mu$ agree on all sufficiently small open subsets of $K.$

Moreover, when $K$ is compact, if $S$ and $T$ are both symmetric rank $r$ tensor fields on $K,$ then to see that $\hat{S}=\hat{T},$ by Theorem \ref{timelikeobserverprinciple}, it suffices to know that $\hat{S}v^r=\hat{T}v^r$ for every timelike smooth vector field on $K.$  In particular, if $\hat{S}v^r=0$ for every timelike vector field $v$ on $K,$ then $\hat{S}=0.$  But, if $\hat{S}=0,$ then it follows that $S=0,$ since now we are assuming $S$ is an actual tensor field.  For if $S \neq 0,$  then we can choose a  function $f:K \lra [0,1] \subset \bR$ which vanishes outside a small open neighborhood $W_f$ of $m$ but with $f(m)=1.$  We can then choose smooth vector fields $w_1,w_2,...,w_r$ on $K$ with $S(m)(w_1(m),w_2(m),...,w_r(m)) \neq 0,$ so the function $S(w_1,w_2,...,w_r)$ does not vanish at $m.$  Then with $h=f^r:K \lra [0,1],$ we have $h(m)=1,$ and $h$ vanishes on $K \setminus W_f.$  Moreover,
\begin{equation}\label{bumpfunctionintegral}
0=\hat{S}(fw_1,fw_2,   \cdot \cdot \cdot   ,fw_r)=\int_K hS(w_1,w_2,  \cdot \cdot \cdot  ,w_r) \mu.
\end{equation}
If $S(w_1,w_2,  \cdot \cdot \cdot  ,w_r)$ has a positive value at $m,$ then after making $W_f$ smaller if necessary, we can assume $S(w_1,w_2,  \cdot \cdot \cdot   ,w_r)$ must have a positive value for all points of $W_f$ and this would mean the integral on the right hand side of equation (\ref{bumpfunctionintegral}) would not vanish which would be a contradiction.  If $S(w_1,w_2, \cdot \cdot \cdot  ,w_r)$ is negative at $m,$ then replacing $S$ by $-S$ we again reach the conclusion $S=0.$  We have therefore proven the following corollary.
\begin{corollary}\label{integraltensorfieldobserverprinciple}
Suppose that $K$ is a compact subset of $M$ and $S$ is a continuous symmetric tensor field on $K$ of rank $r.$  If the monomial form of the multlinear map $\int_K S \mu$ vanishes on all smooth timelike vector fields, then $S=0.$
\end{corollary}

For instance, if $T$ is an energy momentum stress tensor field on the compact subset $K$ of $M$ and if $\int_K T(v,v) \mu=0$ for every continuous timelike vector field $v,$ then $T=0.$  Of course, the corollary can be proven using the above technique together with the special case of Theorem \ref{timelikeobserverprinciple} in finite dimensions, but our development makes it clear that there are more general applications possible where the purely finite dimensional argument would not suffice.  For instance, if $\lambda$ is a smooth 1-form and $S(u,v)=\lambda(\nabla_u v),$ then $S$ is a tensor operator which is not a tensor field.  If $\int S \mu$ vanishes in this case for all sufficiently small open sets, then this would mean that $\lambda(\nabla_u v)=0$ for all pairs of vector fields $u,v.$  If the monomial form of $\int S \mu$ vanishes for all sufficiently small open sets, then $Sym(S)=0,$ and therefore $\lambda(\nabla_u v)+\lambda(\nabla_v u)=0,$ for all pairs of vector fields $u,v.$

Let us define a {\bf local observer field} to be a smooth timelike unit vector field defined on an open subset of spacetime, $M.$  We say that $u$ is a local observer field at the point $p$ of $M$ if $p$ belongs to the  domain of $u.$  If $S$ is a continuous tensor field on $M$ of rank $r,$ and if $u$ is a local observer field with domain $D= $domain$(u)$ then define

$$\int S \mu [u^{(r)}]=\int_D Su^{(r)} \mu,~D=\mbox{domain}(u).$$

\begin{proposition}\label{better}
Suppose $A$ is a subset of $M$ and that for each local observer field $u$ at a point of $A$ we have $\int S \mu [u^{(r)}]=0.$  Then $S$ vanishes on $A.$
\end{proposition}

\begin{proof}
If $p_0$ is a point of $A,$ and  $u_0$ is a timelike unit vector at $p_0,$ it suffices to show that $Su^{(r)}=0,$ by the observer principle.  If not, and if $D$ is the domain of $u,$ then after replacing $D$ by a possibly smaller open subset $V$ with $p_0$ in $V,$ we can assume that $Su^{(r)}$ is never zero on $V,$ which means it must have constant sign, always positive or always negative, on $V.$  Letting $v$ denote the observer field obtained by restricting $u$ to $V,$ we then have $\int S \mu [v^{(r)}] \neq 0,$ a contradiction.

\end{proof}

\section{\bf SPACETIME ENERGY MOMENTUM}

In \cite{C&D2}, the concept of {\bf Spacetime Energy Momentum}  is defined as
$$E=  \frac{1}{4 \pi G}\int \mbox{\bf Ricci } \mu,$$
where $\mu$ is the spacetime metric volume form, and of course {\bf Ricci} is the Ricci curvature tensor of spacetime due to the metric.  By Corollary \ref{integraltensorfieldobserverprinciple}, we see that if the spacetime energy momentum monomial form vanishes on all smooth timelike vector fields on a compact subset of spacetime, then the Ricci tensor itself must vanish everywhere on that compact subset.  As well, by \ref{better}, we see that if the spacetime energy momentum vanishes on every local observer field at points of $A \subset M,$ then by \ref{better}, the Ricci tensor must vanish at all points of $A.$  Alternately, we can think of this as saying that if the Ricci tensor fails to vanish, then some local field of observers must see some spacetime energy.

We can note here, that many definitions of quasi-local energy and momentum also involve integrals of the form discussed here, so that similar remarks would apply to those integrals.  Some of these quasi-local energy momentum definitions have the property that their vanishing implies spacetime is flat.  This is a violation of the idea of the "no prior geometry" concept of general relativity \cite{MTW}.  To say that the energy momentum inside a certain region is zero should imply no more than Riccci flatness, according to the Einstein equation itself.  The full solution for the spacetime with the metric depends on the Einstein equation as well as the boundary conditions, and the boundary conditions should not be considered to contain energy.


\end{document}